\documentclass[journal,print]{ieeecolor}
\usepackage{generic}
\usepackage{cite}
\usepackage{amsmath,amssymb,amsfonts}
\usepackage{mathtools}
\usepackage{mathrsfs}
\usepackage{nccmath}
\usepackage{bbm}
\usepackage{algorithm}
\usepackage{algorithmic}
\usepackage{graphicx}
\usepackage{textcomp}
\usepackage{stfloats}
\usepackage{svg}
\usepackage{xcolor}
\usepackage{adjustbox}
\usepackage{comment}
\usepackage{float}
\usepackage{dsfont}
\usepackage{lipsum}
\usepackage{relsize}
\usepackage{enumerate}
\usepackage{nicematrix,tikz}
\usepackage{arydshln}
\usepackage{balance}
\usepackage{hyperref}

\hypersetup{
    colorlinks=true,
    linkcolor=blue,
    filecolor=magenta,      
    urlcolor=cyan,
    citecolor = red,
    pdfpagemode=FullScreen,
    }

\newsavebox\CBox
\def\textBF#1{\sbox\CBox{#1}\resizebox{\wd\CBox}{\ht\CBox}{\textbf{#1}}}

\newcommand{\Va}{{\mathcal V_{\rm abs}}}
\newcommand{\Vt}{{\mathcal V_{\rm trs}}}
\newcommand{\Vtpt}{{\mathcal V^{t+1}_{\rm trs}}}

\newcommand{\Xt}{{\mathcal X_{x_T =j}^{T=\tau}}}

\newcommand{\cB}{\mathcal B}
\newcommand{\cA}{\mathcal A}
\newcommand{\cC}{{\mathcal C}}

\newtheorem{theorem}{Theorem}[section]

\newtheorem{proposition}{Proposition}[section]

\newtheorem{remark}{Remark}[section]
\newtheorem{problem}{Problem}[section]

\usepackage{multicol}

\begin{document}

\title{Schr\"odinger's control and estimation paradigm\\
with spatio-temporal distributions on graphs}
\author{Asmaa Eldesoukey and Tryphon T. Georgiou
\thanks{This research has been supported in part by NSF under ECCS-2347357, in part by AFOSR under FA9550-24-1-0278, and in part by ARO under W911NF-22-1-292. }
\thanks{The authors are with the Department of Mechanical and Aerospace Engineering, University of California, Irvine, CA; aeldesou@uci.edu, tryphon@uci.edu}}

\maketitle
\begin{abstract}
The problem of reconciling a prior probability law on paths with data was introduced by E. Schr\"odinger in 1931/32. It represents an early formulation of a maximum likelihood problem. This specific formulation can also be seen as the control problem to modify the law of a diffusion process so as to match specifications on marginal distributions at given times.
Thereby, in recent years, this so-called {\em Schr\"odinger's bridge problem} has been at the center of the uncertainty control development.
However, an understudied facet of this program has been to address uncertainty in {\em space} (state) and {\em time}, modeling the effect of tasks being completed contingent on meeting a certain condition at some random time instead of imposing specifications at fixed times.
The present work is a study to extend Schr\"odinger's paradigm on such an issue, and herein, it is tackled in the context of random walks on directed graphs.
Specifically, we study the case where one marginal is the initial probability distribution on a Markov chain, while others are marginals of stopping (first-arrival) times at absorbing states, signifying completion of tasks. We show when the prior law on paths is Markov, a Markov policy is once again optimal to satisfy those marginal constraints with respect to a likelihood cost following Schr\"odinger's dictum. Based on this, we present the mathematical formulation involving a {\em Sinkhorn}-type iteration to construct the optimal probability law on paths matching the spatio-temporal marginals.
\end{abstract}


\section{Introduction} \label{sec:intro}

In a 1931/32 study \cite{Sch31, Sch32}, E.~Schr\"odinger asked for the most likely evolution of particles between two points in time where their distributions are observed.
In this, he laid out elements of a {\em large deviations} theory for the first time. Indeed, at a time when much of probability theory was in its infancy, and H.~Cram\'er's and I. N.~Sanov's theorems \cite{cramer1938nouveau,sanov1958probability} were still a few years away, Schr\"odinger single-handedly formulated and solved a maximum likelihood problem that is now known as the {\em Schr\"odinger's bridge problem} (SBP)\cite{chen2021stochastic}.

Schr\"odinger's paradigm was slow to make inroads into probability theory until about the 1970s with the work of B. Jamison, H. F\"ollmer and others  \cite{jamison1974reciprocal,jamison1975markov,follmer1988random,follmer1988large,csiszar1998method}. Since then, it has become an integral part of stochastic control, fueling the recent fast development of uncertainty control -- a discipline that focuses on regulating state-related uncertainties of stochastic dynamical systems during operation as well as at terminal times, see \cite{todorov2009efficient, vladimirov2015state,  chen2016robust, wang2021survey, chen2021controlling, chetrite2021schrodinger}.
 Notably, \cite{dai1991stochastic} linked the minimum-energy stochastic control problem to the SBP using Fleming's logarithmic transformation \cite{fleming2005logarithmic}. Furthermore, through the body of work in \cite{beghi1996relative,chen2015optimal, chen2015optimalII, chen2018optimal, chen2016relation}, it became evident that steering a stochastic linear system, even of singular noise intensity, between two endpoint Gaussian distributions is optimally achievable via a {\em state-feedback} controller utilizing the SBP formalism. This optimal controller typically has a closed-form expression with a feedback gain that can be obtained by solving two differential Riccati or Lyapunov equations that are nonlinearly coupled at the two endpoints.
In contrast to this earlier literature, where the control purpose is to bring the system to the vicinity of target states at specified times, we are interested in regulating the times when target states are reached.
Thereby, we investigate a different angle to  Schr\"odinger's dictum in which we now seek to reconcile stopping-time (first-passage-time) marginals instead of state-marginal distributions at predetermined times.

 Noticeably, first-passage-time statistics grew indispensable in modeling various processes, to wit, the integrate-and-fire neuron models \cite{redner2001guide}, bacterial steering via switching \cite{morse2015flagellar},  ovarian follicles activating menopause and proliferation \cite{lawley2023slowest,clement2021stochastic}, chemical transport in active molecular processes \cite{neri2017statistics},  stochastic delays in chemical reactions \cite{barrio2013reduction} and many others \cite{chou2014first, iyer2016first,castro2018first,khudabukhsh2020incorporating}. 
A natural next step is to regulate the first-passage temporal statistics via suitable control.
For instance, in congestion control, specifying, e.g., a uniform first-arrival-time distribution of agents at the entrance of a site may be of practical use. In another instance, in landing a space probe, the distribution at landing is of essence, whereas the precise landing time can be, ipso facto, {\em random} due to possibly inescapable stochastic disturbances along the trajectory.
Interestingly,  F. Baudoin \cite{baudoin2002conditioned}, motivated by applications to finance, and
later on, C. Monthus and A. Mazzolo \cite{monthus2022conditioned} studied the conditioning of an It\^o diffusion, starting from a Dirac, to match a first-hitting-time density.

In the present work, we explore such an issue for a controlled random walk on a network. Time and space are discrete, namely, values for the time index $\tau$ belong to some specified finite window $[1,t] \subset \mathbb N$, while the position variable $X_\tau$ of a random walker takes values in the vertex set $\mathcal V$ of a finite directed graph $\mathcal G$ and the set $\mathcal V$ contains absorbing vertices.
As it turns out, our control action entails altering given prior transition probabilities between the vertices (states) to regulate the first-arrival times at the absorbing states over the indicated time window.  
As we restrict the present work to a discrete setting, a congestion control example is presented to numerically demonstrate our approach, while the afore-described space probe problem is spared for future work to a more specialized discussion. 

At a closer look, we identify the control cost with the value of a likelihood functional in observing a collection of random walkers as they traverse the network to match specifications on first arrivals. In that, we frame the control problem as a large deviation problem, in the spirit of Schr\"odinger's approach.
Remarkably, the solution is a Markov law on paths,  or equivalently, the corresponding control action at every point in time is independent of past information given the present state as in original SBP \cite{georgiou2015positive,chen2016robust,chen2016relation} but for different reasons. That is, the controller works to reweigh prior Markov transition kernels.

The structure of the paper is as follows. In Section \ref{sec:martingale}, we bring our setting closer to the reader by casting a {\em De Moivre's martingale} example in an SBP but with stopping-time marginals. In Section \ref{sec:atypical}, we introduce notation and preliminaries. In Section \ref{sec:markov}, we delve into an important quality of the original SBP and, interestingly, of the variant with stopping times, that the optimal posterior law on paths preserves the Markovianity of the prior process in both problems. Explicit construction of the solution to modify the prior so as to match specifications on  stopping probabilities is provided in
Section \ref{sec:SBP}. 
 An outlook on the regularized transport on graphs is in Section \ref{sec:duality}. Section \ref{sec:examples} provides academic examples, including De Moivre's martingale, to elucidate the application of the framework. The paper concludes with a brief discussion in Section \ref{sec:conclusions} on generalizations and potential future directions.

\section{De Moivre's Martingale} \label{sec:martingale}

We begin our exposition with a motivating example, De Moivre's martingale, which models a betting game with two stopping conditions \cite{grimmett2020probability}. The betting problem involves gamblers reaching a set amount before terminating the gambling game. This model problem distinguishes between two possible outcomes: success in reaching the goal or ruin, where all the capital is lost. Now, imagine that we observe a series of game rounds. In each round, one either wins the bet and increases their current wealth by one token or loses with the opposite effect.
Assume the game begins with a set of players having wealth in a finite set $\Vt=\{1,\ldots,M-1\}$ of values for the corresponding number of tokens according to an initial probability distribution $\hat \mu_0$. If a player loses all the capital or achieves a cap of cumulative $M$ tokens in any number of rounds, they discontinue the game. We presume that our prior law corresponds to a fair game. 

Over a discrete window of time $[1,t] \subset \mathbb N$, we observe the distribution of players that terminate the game for the above reasons, reaching the goal of $M$ tokens or being ruined by hitting the lower bound $0$. These two distributions are denoted by $\hat \nu^{\rm win}(\tau)$ and $\hat \nu^{\rm ruin}(\tau)$, for $\tau \in [1,t]$, respectively. However, these observed temporal marginals may not be consistent with the prior law, leading us to suspect foul play. Thus, we seek to identify the {\em most likely} perturbation of the prior law that may explain the observed marginals and point to suspect times when the foul play may have taken place. 

We model this setting via a Markov chain with two absorbing states representing $0$ and $M$ tokens. This Markov chain has a total of $\vert \Vt \vert +2$ states\footnote{$\vert\cdot \vert $ stands for the cardinality of a set.}. The problem being addressed is to determine the most likely law on trajectories beginning with the initial probability distribution $\hat \mu_0$ with support in $\Vt$ and gives rise to the recorded marginals $ \hat \nu^{\rm win}$ and $\hat \nu^{\rm ruin}$. Fig. \ref{fig:fig1} exemplifies this problem. This example is revisited at the end of the paper after a framework in the spirit of Schr\"odinger for matching spatio-temporal marginals is in place.

\begin{figure}[t]
    \centering
\includegraphics[width= \columnwidth]{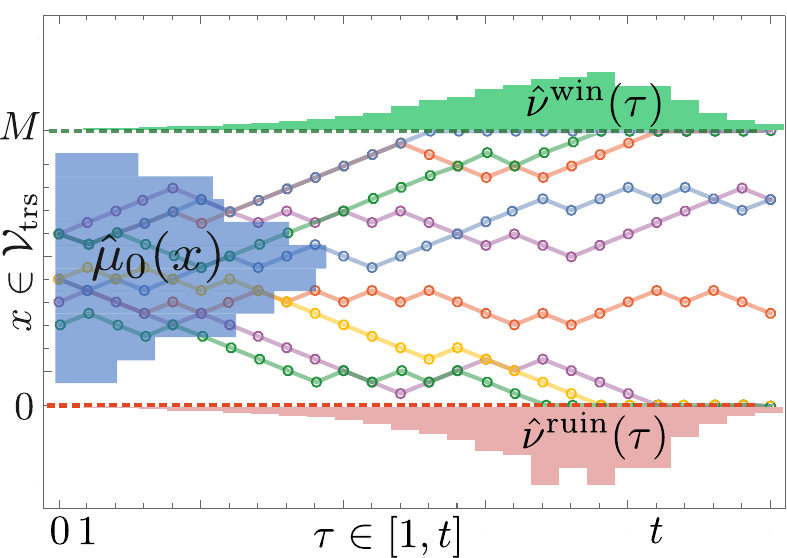}
   \caption{Abstract illustration of De Moivre's Martingale problem setting -- matching the marginal $\hat \mu_0$ (spatial) and the two marginals $\hat \nu^{\rm win}$ and $\hat \nu^{\rm ruin}$ (temporal). The requirement is to find the most likely law on paths.}
    \label{fig:fig1}
\end{figure}

\section{Notation and Rudiments on Large Deviations}\label{sec:atypical}

Consider a directed finite graph $\mathcal G=(\mathcal V,\mathcal E)$, with vertex set $\mathcal V$ and edge set $\mathcal E$. We use $x,y\in\mathcal V$ to denote vertices and often adjoin time $\tau \in\mathbb Z_+ $ as a subscript as in $x_\tau\in\mathcal V$ to combine space-time indexing. We consider a random walk $(X_\tau)_{\tau\in[0,t]}$ on $\mathcal G$ as our process where $[0,t]\subset\mathbb Z_+ $ and $t$ is a bound on times we consider. We denote by $\Va$ the set containing all {\em absorbing} vertices (states) and by $\Vt:=\mathcal V \setminus \Va$ its complement of {\em transient} vertices. Throughout, we set $|\Va|=m$ and  $|\Vt|=n$.

Thus, the sample space $\Omega$ we consider for the paths of the random walk is discrete and finite. We typically use the symbols $Q,P$ to denote probability laws (measures) on sample paths and $L_N$ for an empirical law of $N$ samples. We reserve Greek letters such as $\mu,\nu$ for marginal distributions and $\Pi,\varPi$ for transition matrices. 
Vectors and matrices will be indexed when needed as in $\Pi_\tau$ with subscript denoting time. 
The given initial probability vector will be denoted by $\hat \mu_0 \in \mathbb R_+^{n+m}$, where
$\hat \mu_0(x_0) := P(X_0 = x_0)$.   With no loss of generality, we assume that the support of $\hat \mu_0$ is in $\Vt$, that is, $\hat \mu_0(\Va) = 0$ when we consider the SBP with stopping times.
 Next, we consider the time when the random walk first arrives at $\Va$ as 
 \begin{align*}
T := \inf \{ \tau > 0 \mid X_\tau \in \Va \}.
\end{align*}

This variable can represent the instant of task completion and is random. It is a {\em stopping time}, and reaching any absorbing vertex is the stopping rule\footnote{Since we assume no mass is stationed initially at the absorbing vertices, the variable $T$ takes values strictly greater than $0$.}. We assume given marginals of stopping times, each typically denoted by $\hat \nu_j \in \mathbb R_+^{t}$ and with support in $[1,t]$ and the subscript $j$ denoting a certain absorbing site in $\Va$.
We will be concerned with suitably adjusting a given prior law  $Q$ to match the initial probability vector $\hat \mu_0$ and the specified {\em spatio-temporal} set of marginals $\{ \hat \nu_j,j \in \Va\}$.

In general, a probability law $P$ is said to be {\em absolutely continuous with respect to} $Q$, and denoted by $P \ll Q$,
when $Q(\omega)=0$ implies that $P(\omega)=0$ for any
$\omega \in \Omega$. Alternatively, we also say that the support of $Q$ contains the support of $P$.
The \textit{Kullback-Leibler divergence} from $Q$ to $P$, also known as \textit{relative entropy} \cite{kullback1951information,cover1999elements}, is defined by
 \begin{align}
    \mathbb D (P \parallel  Q) := \left\{\begin{array}{ll}  \sum\limits_{\omega \in \Omega} P(\omega) \log 
\displaystyle \frac{P(\omega)}{Q(\omega)}& \text{if }P \ll Q,\\[8 pt]
 \infty &\text{otherwise},
 \end{array}\right.
 \label{KL}
 \end{align}
 and $0 \log 0 = 0$. 
 To our interest,  the relative entropy can quantify the likelihood of observing empirical distributions when sampling from a given law, as expressed by the following fundamental results in the large deviations theory, see also \cite{sanov1958probability,dembo2009large,cover1999elements}.

 If $\mathcal M$ is a Polish space and $I: \mathcal M \to [0,\infty]$ is a lower semi-continuous map, we say that a sequence of probability measures $(Q_N)_{N \geq 0}$ on $\mathcal M$ satisfies a large deviation principle with a rate function $I$ if for every Borel measurable set $\Gamma \subset \mathcal M$,
\begin{align*}
 -\inf_{\upsilon \in \Gamma^{\mathrm{o}}} I(\upsilon) &\leq  \liminf_{N \to \infty} \frac{1}{N} \log Q_N(\Gamma) \leq \\
 &\limsup_{N \to \infty} \frac{1}{N} \log Q_N(\Gamma) \leq -\inf_{\upsilon \in \overline \Gamma} I(\upsilon),
\end{align*}
where $ \Gamma^{\mathrm{o}},\overline{\Gamma}$ are the interior and the closure of $\Gamma$, respectively.
\vspace{1pt}
\begin{theorem}\label{thm:Sanov} (Sanov's Theorem  in Finite Dimensions) Let  $\mathbf X_1, 
\hdots, \mathbf X_N$ be independent and identically distributed (i.i.d.) random vectors taking values in $\mathcal X$
with $|\mathcal X|<\infty$ and law $Q(\mathbf x)$.
For any set $\Gamma$ of probability distributions over $\mathcal X$ that is closed with nonempty interior,
\begin{align*}
   \lim_{N\to \infty} -\frac{1}{N}\log Q^N(L_N \in\Gamma) = \mathbb D(P^\star  \parallel  Q),
\end{align*}
 where $Q^N$ is the joint probability distribution on $\mathcal X^N$ and\footnote{ $Q^N(\mathbf X_1 =\mathbf x_1, \ldots, \mathbf X_N=\mathbf x_N) : = \prod_{i=1}^N Q(\mathbf X_i = x_i), \ \text{for }  \mathbf x_i \in \mathcal X.$} \begin{align*}
    P^\star  := \arg \min_{P \in \Gamma} \mathbb D(P \parallel  Q). \tag*{$\square$}
\end{align*} 
\end{theorem} 
\vspace{4pt}

 Then, Sanov's theorem states the sequence $Q^N(L_N \in \cdot)$ satisfies the large deviation principle with the rate function $\mathbb D(\cdot \parallel Q)$ in the case of the finite alphabet (i.e., on the space of $\vert \mathcal X \vert$-dimensional probability vectors). Intuitively, the theorem reveals that the likelihood of observing the empirical distributions $L_N \in \Gamma$  under a prior law $Q$ decays exponentially fast with the number of samples $N$
up to a suitable normalizing factor. Thus, the rate function $\mathbb D(\cdot \parallel Q)$ quantifies the likelihood of rare events, and the {\em most likely event}, out of all the rare ones that agree with the neighborhood $\Gamma$ of $L_N$, is the one having least divergence from the prior  $Q$. 

 Curiously, the insights provided by Sanov's result can be traced to the work of Schr\"odinger \cite{Sch31,Sch32,chetrite2021schrodinger}; indeed, in much the same way, Schr\"odinger discretized the path-space trajectories of Brownian particles to arrive at the most likely probability distribution on trajectories that is consistent with observed two endpoint marginals. We call this specific problem setting the {\em classical SBP}.

 Next, we formulate the classical SBP alongside the problem involving stopping-time probabilistic constraints as large deviation problems. Furthermore, we reveal the structure of the solutions in each case. The structure facilitates the choice of appropriate parametrization in search of the solution to our problem starting Section \ref{sec:SBP}.

\section{On Markovian structure of posteriors} \label{sec:markov}

As a prelude to our framework, we lay out the SBP with stopping times preceded by the classical SBP. We also provide a simple derivation of the Markov property in both problems -- when the prior is Markov, so is the most likely posterior. While the  Markovian structure in the classical SBP is standard and underlies the fact that density control can be affected by state feedback \cite[Section 3]{chen2021optimal}, \cite{chen2015optimal}, \cite{chen2021stochastic},
in the present situation where we design for marginals of first-arrival times, the Markovian property in the posterior is not obvious. 
To this end, we consider paths 
\begin{align}\label{eq:paths}
 \mathbf x:=(x_0,x_1,\ldots, x_t)   
\end{align}
over the finite time window $[0,t] \subset \mathbb Z_+ $, taking values in the vertex set $\mathcal V$, and we denote by  $Q(\mathbf x)$ or $Q({x_0,x_1,\ldots,x_t})$ the probability law of a path in $\mathcal V^{t+1}$. 
The law $Q$, or any other law, is Markov,  i.e., it corresponds to a Markov process if and only if it factors as
\begin{equation}\label{eq:Markov}
Q({x_0,x_1,\ldots,x_{t-1},x_t})=q_1({x_0,x_1})\times \cdots \times q_t({x_{t-1},x_t}),    
\end{equation}
for some suitably chosen nonnegative factors. We are not yet concerned with interpreting the factors as conditional probabilities, as this is secondary here. For our purposes, it is only important herein that the law factors as the product of some successive \textit{rank-$2$ tensors}.

\subsection{The Classical SBP}\label{sec:SBsection}

 The classical SBP is described in the discrete setting by the following large deviation problem.
 \vspace{4pt}
\begin{problem} \label{prob:classSBP}  Given a prior law $Q$ on $\mathcal V^{t+1}$ and two endpoint marginals $\hat \mu_0, \hat \mu_t \in \mathbb R_+^{n+m}$ with support in $\mathcal V$, determine
\begin{subequations} \label{eq:SB}
\begin{align}\label{eq:SB1}
    P^\star := \arg  \min_{P \ll Q} \sum_{\mathbf x \in \mathcal V^{t+1}} P(\mathbf x) \log \frac{P(\mathbf x)}{Q(\mathbf x)}, \\
\label{eq:SB2}
\text{subject to }     \sum_{x_1,\ldots,x_t \in \mathcal V} \! \! \! P(x_0,x_1,\ldots, x_{t-1},x_t) &= \hat \mu_0{(x_0)},  \\ 
     \label{eq:SB3}
   \text{and } \sum_{x_0,\ldots, x_{t-1} \in \mathcal V} \! \! \!P(x_0,x_1,\ldots,x_{t-1}, x_t) & = \hat \mu_t{(x_t)}, 
\end{align} 
 for all $x_0,x_t \in \mathcal V$. \hfill{$\square$}
\end{subequations} 
\end{problem}
\vspace{4pt}

Here, and throughout all similar optimization problems, we assume that the conditions are feasible (here \eqref{eq:SB2} and \eqref{eq:SB3}), and thereby, strong duality holds.
Introducing Lagrange multipliers,  the augmented Lagrangian of Problem \ref{prob:classSBP} is
\begin{align*}
\mathcal L_{\rm SBP} (P,\Lambda_0,\Lambda_t) = \! \!  \! \!  \! \! &\sum_{x_0,\ldots,x_t \in \mathcal V} \! \! \bigg\{ P(x_0, \ldots, x_t) 
\log \frac{P(x_0,\ldots, x_t)}{Q(x_0\ldots, x_t)}  \\
&+  \Lambda_0(x_0) \big[P(x_0,\ldots, x_t)-\hat \mu_0(x_0)\big] \\
&+  \Lambda_t(x_t) \big[P(x_0,\ldots, x_t) -\hat \mu_t(x_t) \big]\bigg\}.
\end{align*}
$\sum_{x_0,\ldots,x_t \in \mathcal V} \Lambda_0(x_0)P(x_0,\ldots, x_t)$ can be rewritten as
\begin{equation*}\label{eq:convenient}
\sum_{i_0\in\mathcal V}\sum_{x_0,\ldots,x_t\in \mathcal V}\Lambda_0(i_0)P(x_0,\ldots, x_t) \mathds 1_{\{x_0=i_0\}}(x_0,\ldots, x_t),
\end{equation*}
and similarly for the other summation involving $\Lambda_t$, where  $\mathds 1_{\{x=\alpha\}}: \mathcal V^{t+1} \to \{0,1\}$ denotes the characteristic (indicator) function which equals unity if and only if $x=\alpha$. Then, the minimizer takes the form
    \begin{align} \label{eq:Pstar}
    P^\star (\mathbf x) = Q(\mathbf x) & \times e^{-1}  \times \exp {\bigg(-\sum_{i_0\in\mathcal V}\Lambda_0(i_0)\mathds 1_{\{x_0=i_0\}}(\mathbf x)\bigg)} \nonumber \\
    & \times \exp{\bigg(-\sum_{i_t\in\mathcal V}\Lambda_t(i_t)\mathds 1_{\{x_t=i_t\}}(\mathbf x)\bigg)}.
\end{align}

Now, note that 
\begin{align*}
\exp{\bigg(-\sum_{i_0 \in \mathcal V}\Lambda_0(i_0)\mathds 1_{\{x_0=i_0\}}(\mathbf x)\bigg)}  = 
e^{-\Lambda_0(x_0)},
\end{align*}
and similarly,
\begin{align*}
\exp{\bigg( -\sum_{i_t \in \mathcal V}\Lambda_t(i_t)\mathds 1_{\{x_t=i_t\}}(\mathbf x) \bigg)}  = 
e^{-\Lambda_t(x_t)}.
\end{align*}
Thus, $P^\star (x_0,x_1,\ldots,x_t)$ is expressed as the product
\begin{align}\label{eq:followup}
Q(x_0,x_1,\ldots,x_t)\times \underbrace{f_1(x_0,x_1) \times \cdots \times f_{t}(x_{t-1},x_t)}_{F(x_0,x_1,\ldots, x_t)},
\end{align}
where $f_1(x_0,x_1)=e^{-\Lambda_0(x_0)-1}$ and $f_{t}(x_{t-1},x_t)=e^{-\Lambda_t(x_t)}$ are rank-$1$ tensors (since their values depend on one index) and all other factors are trivially equal to 1.
Our choice to display the factorization of the correction tensor $F$ in full draws a parallel and contrast to the next subsection. This provides a rather transparent derivation of the following well-known result that can be traced to its continuous counterpart \cite{Sch31,Sch32,chen2016relation} pertaining to diffusion processes.

\vspace{4pt}
\begin{theorem}\label{thm:firstmainthm}
    Assume that $Q$ is Markov and that Problem \ref{prob:classSBP} is feasible. Then, the unique (due to convexity) corresponding solution $P^\star $ is also Markov. \hfill{$\square$}
\end{theorem} 
\vspace{4pt}

\begin{proof}
Since the constraints are affine and the objective is strictly convex, strong duality holds, and the solution is in the form indicated in \eqref{eq:Pstar}, and therefore, it also satisfies the follow-up expression \eqref{eq:followup}.
    The product of tensors in the factorizations
    $q_1({x_0,x_1}) \times \cdots \times q_t({x_{t-1},x_t})$ and
    $f_1(x_0,x_1) \times \cdots \times f_t(x_{t-1},x_t)$ is again a tensor of the same form, i.e., a product of successive rank-$2$ tensors. Hence, the probability law $P^\star $ is Markov.
\end{proof}

\subsection{The SBP with Stopping Times}

The SBP with stopping times aims to reconcile marginals for random variables that are of a different nature. Specifically, these are marginals of $X_0$ and of the stopping times at the absorbing states $j\in \Va$. 
 Accordingly, we now seek a solution to the following problem.
\vspace{4pt}
\begin{problem}\label{prob:SBwithST}
 Given a prior law $Q$ on $\mathcal V^{t+1}$, determine a law $P^\star $ that satisfies \eqref{eq:SB1},
matches the prescribed starting probability distribution  $\hat \mu_0$ as in \eqref{eq:SB2}  with $\hat \mu_0(\Va) =0$ and allocates a pre-specified amount of probability mass at the absorbing states on the interval $[1,t]$ (see Eq.\ \eqref{eq:cumulative} below).  \hfill{$\square$}
\end{problem}
\vspace{4pt}
Thus, to contrast with classical SBPs, the terminal constraint \eqref{eq:SB3} is replaced by a running constraint on the absorbed mass in what follows.
Tautologically, the probability of paths being stopped at a site $j$ and at an instant that is no later than $\tau$ equals the cumulative probability of stopped walkers at the site $j$ up to the instant $\tau$. That is, for all $\tau \in [1,t]$ and $j \in \Va$, we have
\begin{align}\label{eq:cumulative}\tag{\ref{eq:SB}d}
    \sum_{\mathbf x \in \mathcal V^{t+1}} P({\mathbf x}) \mathds 1_{\{x_\tau=j\}}(\mathbf x) = \sum_{s=1}^\tau\hat \nu_j(s).
\end{align}

The augmented Lagrangian of  Problem \ref{prob:SBwithST} is
\begin{align*}
& \mathcal L_{\rm SBP}^{\rm stop}
(P,\Lambda_0,\Lambda) = \! \! \! \! \!  \! \!  \sum_{x_0,\ldots,x_t \in \mathcal V} \! \!  \bigg\{ P(x_0,\ldots, x_t)
\log \frac{P(x_0,\ldots, x_t)}{Q(x_0,\ldots, x_t)} \\
&  + \Lambda_0(x_0) P({x_0,\ldots,x_t})- \Lambda_0(x_0)\hat\mu_0(x_0) \\
&+  \sum_{j \in \Va} \!\sum_{\tau=1}^t 
\Lambda(\tau,j) \mathds 1_{\{x_\tau=j\}}({x_0,\ldots,x_t})P({x_0,\ldots,x_t}) \\
&  - \sum_{j \in \Va} \!   \sum_{\tau=1}^t 
\Lambda(\tau,j) \sum_{s=0}^\tau \hat \nu_j(s)\!\bigg\}.
\end{align*}
The Lagrange multipliers $\Lambda_0$ and $\Lambda$ enforce the constraints \eqref{eq:SB2} and \eqref{eq:cumulative}, respectively.  The first-order optimality condition reveals the optimizer's form as
\begin{align}
    P^\star (\mathbf x) = & Q(\mathbf x)  \times e^{-1} \times \exp{\bigg(-\sum_{i\in\mathcal V}\Lambda_0(i)\mathds 1_{\{x_0=i\}}(\mathbf x)\bigg)} \nonumber \\
    & \times \exp{\bigg(- \sum_{j \in \Va}\sum_{\tau=1}^t 
\Lambda(\tau,j) \mathds 1_{\{x_\tau=j\}}(\mathbf x)\bigg)},
    \label{eq:convenient3}
\end{align}
and this brings us to the following statement.

\vspace{4pt}
\begin{theorem} \label{thm:thm2}
    Assume that $Q$ is Markov and that  Problem \ref{prob:SBwithST} is feasible. Then, the unique (due to convexity) minimizer $P^\star $ is also Markov and given by \eqref{eq:convenient3}. \hfill{$\square$}
\end{theorem}
\vspace{4pt}

\begin{proof}
The solution $P^\star $ is in the form \eqref{eq:convenient3}; thus, it is the product of $Q$ with tensors that depend on the Lagrange multipliers. We only need to verify that each tensor factors into a succession of rank-2 tensors. In such case, the Markovian character \eqref{eq:Markov} of the prior will be preserved in $P^\star $.
As in the classical SBP case in Subsection \ref{sec:SBsection}, the term that involves $\Lambda_0$, namely,
    \begin{align*}
       \exp{\bigg(-\sum_{i\in\mathcal V}\Lambda_0(i) \mathds 1_{\{x_0=i\}}(\mathbf x)\bigg)} = e^{-\Lambda_0(x_0)},
    \end{align*}
    is a rank-$1$ tensor as before, and $e^{-1}$ is just a scaling factor. In turn, their contributions are of the required factored form. The new element here is
 the last factor in \eqref{eq:convenient3}, which we write as
    \begin{align*}
        \exp{\bigg(\! \!- \! \!\sum_{j \in \Va}\! \sum_{\tau=1}^t 
\Lambda(\tau,j) \mathds 1_{\{x_\tau=j\}}(\mathbf x)\bigg)}=\! \!\prod_{j \in \Va}\! \! {F^j(x_0,\ldots, x_t)},
    \end{align*}
    for some factors $F^j$, one for each absorbing state, of the form
    \begin{align*} F^j(x_0,\ldots,x_t)&= f_1^j(x_0,x_1) \times \cdots \times f_t^j(x_{t-1},x_t),
    \end{align*}
    where this next layer of factors are all rank-$1$ tensors, scaling only a corresponding direction $x_\tau$ as follows.
    \begin{align*}
f_\tau^j(x_{\tau-1},x_\tau)&=\exp{\big(-\Lambda (\tau,j)\mathds 1_{\{x_\tau=j\}}(\mathbf x)\big)}\\
&=\left\{\begin{array}{ll}
e^{-\Lambda (\tau,j)} & \text{ if }x_\tau=j,\\
1 & \text{ otherwise}.
\end{array}\right.
    \end{align*}
 In other words, 
  $f^j_\tau$ can be viewed as a matrix with all entries equal to $1$ except for the $j$-th column that is scaled\footnote{Once again, although the value of $f_\tau^j$ only depends on one coordinate, we choose to display it as a $2$-dimensional tensor to highlight its Markov-like factorization. An alternative equivalent factorization is possible where adjacent factors split the contribution of $e^{-\Lambda(\tau,j)}$ in half between them, with one having a column scaled by $e^{-\Lambda(\tau,j)/2}$ and the other a corresponding row by the same scaling.} by $e^{-\Lambda (\tau,j)}$. 
  This product structure preserves the Markovian property, which completes the proof.
\end{proof}

    \vspace{4pt}
\begin{remark}
    Theorem \ref{thm:thm2} holds when 
    \begin{align*} 
       \sum_{j \in \Va} \sum_{\tau=1}^t  \hat \nu_j(\tau) \leq 1,
    \end{align*}
     implying that a subset (strict or not) of the random walkers arrive at the absorbing destinations by the end of the time interval $[1,t]$. Equality is not necessary for the validity of the theorem, and when strict inequality holds, the remaining probability mass is distributed on the transient states. In other words, a subpopulation of random walkers have not reached an absorbing site and are still marching inside the network. In this case, the posterior $P^\star$ is to estimate the occupancy of these transient states at time $t$, which signifies the end of the interval over which observations are made over the absorbing states. \hfill{$\square$}
\end{remark}
\vspace{4pt}

 To sum up, this section showed that when the prior law on paths is Markov, the {\em most likely} posterior is also Markov for both types of problems: the classical SBP and the SBP with stopping times. 
This justifies the parametrization used next in the search for the optimal solution.  

\section{Transition rates \& spatio-temporal control}\label{sec:SBP}

It is a welcomed surprise that an update of the prior law, compatible with specifications on the initial probability vector and cumulative stopping-time probabilities at a collection of sites, shares the Markovian character of the prior as established by Theorem \ref{thm:thm2}. From the estimation perspective, this means that given a prior random walk on the graph $\mathcal G$ with a corresponding law $Q$ on sample paths and spatio-temporal marginals,  these marginals are optimally realizable via a posterior random walk with the corresponding law $P^\star$ in \eqref{eq:convenient3}. This update of the prior law can be interpreted as both the solution of a maximum likelihood problem \`a la Schr\"odinger, as well as the solution of a control problem to \textit{steer} stochastic flows on the graph to satisfy spatio-temporal marginals that are suitably specified.

The salient feature of our approach is that the prior law on paths, realized by a sequence of transition probabilities since it is Markov, provides a reference for the stochastic evolution on the graph. 
Then, the most likely posterior $P^\star$ in \eqref{eq:convenient3} assigns accordant amounts of probability mass on sets of space-time indices and is closest to the given prior law  $Q$ in the relative entropy.
This posterior will be obtained by adjusting the prior transition probabilities where such an adjustment at each time instant represents the desired {\em control action}. 
 At first glance, it may not have been clear that steering stochastic flows on the graph toward stopping-time marginals does not require dependency on, e.g., the full path or partial history of states, which would have implied that the required control action, say, at instant $k$, is some function of, e.g., $(x_0,x_1, \ldots, x_k)$. However, Theorem \ref{thm:thm2} further implies our steering control task is to specify a suitable protocol for obtaining new transition probabilities via reweighing ({\em rescaling}) prior Markov transitions analogous to a state-feedback input to diffusion processes. This is done next.

\subsection{Problem Formulation}

 In this section and the following, we take the set $[1,m] \subset \mathbb N$ to index elements in $\Va$ and $[m+1,m+n] \subset \mathbb N$ to index elements in $\Vt$ for convenience. We set the prior Markov law on $\mathcal V^{t+1}$ as
\begin{subequations} \label{eq:priorlaw}
\begin{align} \label{eq:priorlaw1}
\! \!    Q(\mathbf x) = \mu_0(x_0) \Pi_1(x_0,x_1) \Pi_2(x_1, x_2) \cdots \Pi_t( x_{t-1},x_t),
\end{align}
where $\mu_0$ is not necessarily equal to $\hat \mu_0$ and $\Pi_\tau$ denotes a transition matrix on $\mathcal V \times \mathcal V$. 
 Namely, 
\begin{align*}
    \Pi_\tau(x_{\tau-1},x_\tau): = Q(X_\tau = x_\tau \mid X_{\tau-1}=x_{\tau-1})
\end{align*} and $\Pi_\tau$ is a row-stochastic matrix indexed by time\footnote{Our choice to consider a time-varying kernel facilitates comparison with the posterior that typically has a time-varying structure.} $\tau$. We also set  
\begin{align} \label{eq:priormat}
    \Pi_\tau := \begin{bmatrix}
        I & \mathbf 0  \\
        B_\tau & A_\tau
    \end{bmatrix},
\end{align}
\end{subequations}
where $I$ and $\mathbf 0$ denote the identity and the zero matrices\footnote{Throughout, the dimensions of $I$ and $\mathbf 0$ will be implied from the context; e.g., here $I=I_{m}$ and $\mathbf 0=\mathbf 0_{m \times n}$.}, respectively, and where
\begin{align*}
    B_\tau(x,y) := \Pi_\tau(x,y), \ \ \text{ and }  \ \ A_\tau(x,z) := \Pi_\tau(x,z),
\end{align*}
 for all $x,z \in \Vt$ and $y \in \Va$.  Once again, from Theorem \ref{thm:thm2}, we know that the solution to Problem \ref{prob:SBwithST} is Markov. Hence, our aim now is to construct explicitly the optimal transition kernel $(\Pi_\tau^\star)_{\tau \in [1,t]}$ realizing the most likely law $P^\star $ with $\hat \mu_0$  and $\{\hat \nu_1, \hdots, \hat \nu_m\}$ as marginals. Hence, 
 our task is to search for the optimal $P^\star$ in all laws of the form
 \begin{subequations}\label{eq:MLaw}
    \begin{align}\label{eq:MLaw1}
\! \!   P(\mathbf x) =  \hat \mu_0(x_0) \hat \Pi_1 (x_0,x_1) \hat \Pi_2 (x_1,x_2) \cdots \hat \Pi_t{( x_{t-1}},x_t),
\end{align}
where each $\hat \Pi_\tau$ is also row stochastic, and similar to $\Pi_\tau$, 
\begin{align} \label{eq:postmat}
    \hat \Pi_\tau := \begin{bmatrix}
        I & \mathbf 0  \\
       \hat B_{\tau} & \hat A_{\tau}
    \end{bmatrix}
\end{align}
 is a transition matrix on $\mathcal V \times \mathcal V$ with
\begin{align*}
   \hat B_\tau (x,y) := \hat \Pi_\tau(x,y), \ \ \text{ and }  \ \   \hat A_\tau (x,z) := \hat \Pi_\tau(x,z),
\end{align*}
\end{subequations}
for all $x,z \in \Vt$ and $y \in \Va$. The problem is set as follows.

\vspace{4pt}
 \begin{problem}\label{prob:KL-paths-mult}
     Consider a prior Markov law $Q$ on paths in $\mathcal V^{t+1}$ is given as in \eqref{eq:priorlaw} together with a probability (column) vector $\hat\mu_0  \in \mathbb R_+^{n+m}$ with support in $\Vt\subset \mathcal V$
    and a set of (row) vectors $\{\hat \nu_j  \in \mathbb R_+^{t} \mid j \in  [1,m]\}$
    as marginals of first-arrival (stopping) times at $\Va$ and each with support in $[1,t]$ such that
     \begin{align} \label{eq:ineq}
     \sum_{j=1}^m \sum_{\tau=1}^t\hat \nu_j(\tau)\leq 1.
     \end{align}
     Determine the transition kernel $(\Pi_\tau^\star)_{\tau \in [1,t]}$ for the Markov law
     \begin{subequations}
        \begin{align}\label{eq:likelihoodfn}
       P^\star  := \arg \min_{P \ll Q}   \sum_{\mathbf x \in \mathcal V^{t+1}}  P(\mathbf x) \log \frac{P(\mathbf x)}{Q(\mathbf x)},
        \end{align}
     with $P$ of the form \eqref{eq:MLaw} and subject to 
        \begin{align}
    \sum_{x_1,\ldots,x_t \in \mathcal V} P(x_0,x_1,\ldots, x_t) & = \hat \mu_0(x_0), \label{eq:constr1}\\
  \text{and }  \sum_{\mathbf x \in \Xt}P(\mathbf x) & = \hat \nu_j(\tau) \label{eq:constr2},
        \end{align}
        \end{subequations}
     for any $x_0 \in \mathcal V, \tau \in [1,t], j \in [1,m]$, and
     \begin{align}\label{eq:setofpaths}
     \! \! \! \! \Xt:=\{ \mathbf x \in \mathcal V^{t+1} \mid  x_0,\ldots,x_{\tau-1}\neq j, x_\tau=j\}
     \end{align}
     is the set of paths that first arrive at $j$ at instant $\tau$.   \hfill{$\square$}
     \end{problem}
\vspace{4pt}

 We proceed with technical propositions that help us build the solution. At first, we consider a problem of reconciling marginals that can be partially specified, similar to the case in \eqref{eq:ineq} (Section \ref{sec:I}).  This shall serve as a step to solve Problem \ref{prob:KL-paths-mult}. It is followed by obtaining the form of the optimal Markov kernel $(\Pi_\tau^\star)_{\tau \in [1,t]}$ that builds the optimizer $P^\star$ (Section \ref{sec:II}). We conclude with an explicit verification that the solution we provide is optimal (Section \ref{sec:III}). The explicit construction of the Markov law provides an alternative, independent argument that proves Theorem \ref{thm:thm2}. 
The transition kernel of the optimal Markov protocol can be employed to steer random walkers to fulfill the prescribed marginals $\hat \mu_0$ and $\{\hat \nu_1,\ldots, \hat \nu_m\}$.

\subsection{Building  Space-Time Bridges I:  Reconciling Marginals}\label{sec:I}
 Consider a {\em transition} probability matrix $\varPi \in \mathbb R_+^{n \times (r+\ell)}$, partitioned
as
\begin{align*}
 \varPi =\begin{bmatrix}
    \cB, & \cA
\end{bmatrix},
\end{align*}
where $\cB \in \mathbb R_+^{n \times r}$ and $\cA \in \mathbb R_+^{n \times \ell}$.
The row stochasticity of $\varPi$, $\sum_{y=1}^{r+\ell} \varPi(x,y)=1,$
will be compactly expressed as
$\varPi \mathds 1=\mathds 1$,
where $\mathds 1$ is a (column) vector of $1$'s of suitable size\footnote{Thus, when writing $\varPi\mathds{1}=\mathds{1}$, in the first occurrence $\mathds 1\in \mathbb R^{r+\ell}$, and in the second, $\mathds 1\in \mathbb R^{n}$.}.

Consider now that we are given a marginal {\em probability} (column) vector $\hat\mu \in \mathbb R_+^n$, and another marginal (column) vector $ \hat \nu \in  \mathbb R_+^r$. The latter can be {\em partially specified}. That is, we are given a non-negative vector $\hat \nu$ with $\sum_{y=1}^r \hat \nu(y)\leq 1$. Without loss of generality, we may assume that 
$$\hat \mu^\intercal \cB \neq \hat \nu^\intercal,$$
with  $^\intercal$ being {\em transpose}, and the task is to determine the transition probability matrix $\varPi^\star$, partitioned similarly as 
\begin{align*}\varPi^\star= \begin{bmatrix}
     \cB^\star, & \cA^\star
\end{bmatrix}\end{align*}
so that $$\hat \mu^\intercal \cB^\star  = \hat \nu^\intercal,$$
following Schr\"odinger's dictum. 
Thus, we seek the solution to the following problem.

\vspace{4pt}
\begin{problem}\label{prob:partialSB}
     Given two marginals  $\hat \mu \in \mathbb R_+^n$ and $\hat \nu \in \mathbb R_+^r$ such that $\hat \mu^\intercal \mathds 1 = 1$, and $\hat \nu^\intercal \mathds 1 \leq 1$, determine
    \begin{align*}
      \varPi^\star := {\rm arg}\min_{\hat \varPi \ll \varPi} &
    \sum_{x =1}^n \sum_{y =1}^{r+\ell} \hat \mu(x)\hat \varPi(x,y)\log \frac{\hat \varPi(x,y)}{\varPi(x,y)}, \\
   \text{subject to } \hat \mu^\intercal \hat \cB &= \hat \nu^\intercal,  \ \ \text{and }  \ \ \hat \varPi\mathds{1} =\mathds{1},
    \end{align*}
    where $\hat\varPi= \begin{bmatrix} \hat \cB, & \hat \cA \end{bmatrix}$, $\hat \cB \in \mathbb R_+^{n \times r}$, and $\hat \cA \in \mathbb R_+^{n \times \ell}$.   \hfill{$\square$}
\end{problem}
\vspace{4pt}

 Problem \ref{prob:partialSB} at first glance seems quite abstract; however, it will be used as a step in solving Problem \ref{prob:KL-paths-mult}. In that, we relate the marginals $\hat \mu$ and $\hat \nu$ to our marginals $\hat \mu_0$ and $\{\hat \nu_1, \ldots, \hat \nu_m\}$ and $\varPi$ and $\varPi^\star$ with the respective Markov kernels of $Q$ and $P^\star$. To see this, consider $r=m$, $\ell =n$,
\begin{align*}
    \hat \mu = \hat \mu_0(m+1:m+n)  \  &, \   \hat \nu^\intercal = \big [\hat\nu_1(1),\ldots, \hat\nu_m(1) \big], \\
    \cB = B_1,  \ \ \ \ & \text{and } \ \ \ \  \cA = A_1,
\end{align*}
then Problem \ref{prob:partialSB} is identical to Problem \ref{prob:KL-paths-mult} but for a single time step, i.e., for $t=1$. In subsequent sections, we shall consider the specific case in which we take 
\begin{align*}
\cB=[B_1,\, A_1B_2,\ldots,\,(A_1\cdots A_{t-1})B_t], \ \text{ and } \ \cA= A_1\cdots A_t.
\end{align*}
 Then, it will hold that $r = t m$ and $\ell=n$, but here we let $r$ and $\ell$ be arbitrary for generality. Incidentally, if the prior kernel is time-invariant, then \begin{align*}
 \cB=[B,\, AB,\ldots,\, A^{t-1}B],  \ \text{ and } \ \cA=A^t, \end{align*}
 which are simpler and perhaps seem more familiar.

    \vspace{4pt}
    \begin{remark}
     The problem of reconciling marginals that are partially specified on both ends can be treated similarly. However, for simplicity, we refrain from this additional layer of generality as it will not be needed in the sequel. \hfill{$\square$}
\end{remark}
\vspace{4pt}

Assuming feasibility, we readily obtain the form of the minimizer of Problem \ref{prob:partialSB} by considering the  augmented Lagrangian
\begin{align*}
    \mathcal L =&  \sum_{x=1}^n \hat \mu (x)  \bigg[ \ 
\sum_{y=1}^r \hat \cB(x,y)
\log \frac{\hat \cB(x,y)}{\cB(x,y)} \\
&+ \sum_{y=1}^\ell  \hat \cA(x,y)
\log \frac{\hat \cA(x,y)}{\cA(x,y)} \bigg]\\
&+  \sum_{y=1}^r \lambda(y) \bigg[\sum_{x=1}^n \hat \mu(x) \hat \cB(x,y) - \hat \nu(y)\bigg]   \\
&+ \sum_{x=1}^n d(x) \bigg[\sum_{y=1}^r  \hat \cB(x,y) + \sum_{y=1}^\ell \hat \cA(x,y) -1 \bigg], 
\end{align*}
with  $\lambda$ and $d$ being Lagrange multiplier vectors. 
Setting the partial derivatives with respect to the entries of $\hat \cB$ and $\hat \cA$ to zero gives that the minimizer is such that
\begin{align*}
     \cB^\star (x,y) &=\cB(x,y) e^{\mathlarger{-1-\frac{d(x)}{\hat \mu (x)}}} e^{-\lambda(y)},\\
  \text{and }   \cA^\star(x,y) & =\cA(x,y) e^{\mathlarger{-1-\frac{d(x)}{\hat \mu (x)}}}.
\end{align*}
Thus, the Lagrange multipliers effectively scale the two matrices $\cB$ and $\cA$
by the diagonal matrices (diagonal scaling)
 \begin{align*}
 \mathbf D_0 = {\rm diag}(D_0), \ \ \ \ & \text{ and } \ \ \ \ \mathbf \Lambda={\rm diag}(\Lambda), \\
  \text{with } \ \   D_0(x) := e^{\mathlarger{-1-\frac{d(x)}{\hat \mu (x)}}}, \ \ & \text{ and } \ \  \Lambda(y) :=e^{-\lambda(y)},
\end{align*}
where both $\cB$ and $\cA$ are multiplied on the left by $\mathbf D_0$ and only $\cB$ on the right by $\mathbf \Lambda$. 
Furthermore,  as stated in the following proposition, these diagonal scalings are uniquely determined iteratively by a {\em Sinkhorn}-type algorithm.

\vspace{4pt}
\begin{proposition}\label{prop:partial}
    The solution to Problem \ref{prob:partialSB}, assuming feasibility, is
$  \varPi^\star = \begin{bmatrix}
     \cB^\star, & \cA^\star    \end{bmatrix} $
    such that
    \begin{align}
        \cB^\star=\mathbf D_0 \cB\mathbf \Lambda, \ \ \text{ and }     \ \    \cA^\star = \mathbf D_0 \cA,
    \end{align}
    for some  nonnegative diagonal matrices $\mathbf D_0$ and $\mathbf \Lambda$.
    The column vectors  $D_0 \in \mathbb R_{++}^n$ and $\Lambda \in \mathbb R_+^r$ respectively corresponding to $\mathbf D_0$ and $\mathbf \Lambda$ can be obtained as the
    limits to the following Sinkhorn iteration:
\begin{subequations}\label{eq:Sinkhorn}
\vspace{4pt}
    \begin{align} 
    \Lambda^{(k)} & \mapsto  D_0^{(k)} =  \mathds 1 \oslash (\cB \Lambda^{(k)} + \cA \mathds 1), \\
    D_0^{(k)} & \mapsto \Lambda^{(k+1)}  = \hat \nu \oslash (\cB^\intercal  ({\hat \mu} \odot D_0^{(k)})), 
    \end{align}
    \end{subequations}
    carried out for $k=1,2,\ldots$ and initialized by taking $\Lambda^{(1)}=\mathds 1$, and where  $\odot,\oslash$ denote entry-wise multiplication and division of vectors, respectively. \hfill{$\Box$}
\end{proposition}
\vspace{4pt}
\begin{proof}
    The proof follows from the preceding discussion of the proposition and standard arguments, as in \cite{georgiou2015positive}, \cite[Section 4.2]{peyre2019computational}, noting that map in \eqref{eq:Sinkhorn} comprises either contractive or isometric maps in the Hilbert metric.
\end{proof}

\subsection{Building  Space-Time Bridges II:  Obtaining the Optimal Markov Kernel}\label{sec:II}
 Herein, we utilize the method in the previous subsection to solve the general Problem \ref{prob:KL-paths-mult} by addressing the case where $r = tm$ and $\ell = n$. For that, we now consider the Markov model \eqref{eq:priorlaw} and write the prior probability kernel of the random walk $(X_\tau)_{\tau \in [0,t]}$
transitioning to
\begin{align*}
\underbrace{\Va \times \Va\times \cdots \times \Va}_{\text{stopped process over} \ [1,t]}\times \Vt,
\end{align*}
 given that it starts from  $\Vt$ via a {\em telescopic} expansion of successive transitions:
\begin{align*} 
\begin{bmatrix}
    B_1, & A_1 B_2, &\ldots, & (A_1 \cdots A_{t-1})B_t, & (A_1\cdots A_t)
\end{bmatrix}.
\end{align*}
Thus, e.g., the $(x,y)$-entry of $A_1 B_2$ quantifies the probability that the random walk first arrives at the absorbing vertex $y$ at $\tau=2$  given that it starts at the transient vertex $x$, i.e., this entry equals $Q(X_T = y, T=2 \mid X_0 = x)$, and so on.
Then, $(A_1\cdots A_t)$ is the transition matrix into the transient states for walkers that have avoided being absorbed over the whole window $[1,t]$.

We consider as data  the marginal distribution $\hat \mu$ on $\mathcal \Vt$ for the random walk at time $\tau=0$  constructed from $\hat \mu_0$ such that
\begin{align} \label{eq:mu}
    \hat \mu = \hat \mu_0(m+1:m+n)
\end{align}
and the {\em spatio-temporal} marginal constructed from $\{\hat \nu_1, \ldots, \hat \nu_m\}$ such that
\begin{align}
&\hat \nu^\intercal = \nonumber \\
& \ \ \big[\underbrace{\hat\nu_1(1),\ldots, \hat\nu_m(1)}_{T=1}, \underbrace{\hat\nu_1(2),\ldots,\hat\nu_m(2)}_{T=2}, \ldots, \underbrace{\hat\nu_1(t),\ldots, \hat\nu_m(t)}_{T=t}\big], \label{eq:v}
\end{align}
 The vector $\hat \nu \in \mathbb R_+^{tm}$ encapsulates our information on first-time arrivals at the absorption vertices over the interval $[1,t]$. We then  set $\varPi$ as the telescopic expansion and partition it into two submatrices such that
\begin{align*} 
\varPi= \bigg[
    \underbrace{B_1 , \ \ \ A_1 B_2, \ \ \  \ldots, \ \ \  (A_1 \cdots A_{t-1})B_t}_{\cB}, \ \ \  \underbrace{(A_1\cdots A_t)}_{\cA}
\bigg],
\end{align*}
and apply Proposition \ref{prop:partial} directly while recalling in this case $r = tm$ and $\ell =n$. We deduce that there exists a unique pair of diagonal scaling:  $\mathbf D_0 \in \mathbb R_+^{n \times n}$ on the left and
\begin{align*}
\mathbf \Lambda = {\rm diag}(\mathbf \Lambda_1, \ldots, \mathbf \Lambda_t) \in  \mathbb R_+^{tm \times tm}
\end{align*}
on the right  with $\mathbf \Lambda_1, \ldots, \mathbf \Lambda_t \in \mathbb R_{+}^{m \times m}$ so that
the solution to Problem \ref{prob:partialSB} is of the form
\begin{align}
\varPi^\star =    \mathbf D_0 \varPi \begin{bmatrix}\mathbf \Lambda_1& \mathbf 0 &\ldots &\mathbf 0 &\mathbf0\\
\mathbf 0 & \mathbf \Lambda_2 & \ldots & \mathbf 0& \mathbf 0\\
\vdots &\vdots&\ddots&\vdots& \vdots\\
\mathbf0 & \mathbf 0& \ldots & \mathbf \Lambda_t &\mathbf 0\\
\mathbf 0 & \mathbf 0 & \ldots & \mathbf 0 & I
\end{bmatrix}\label{eq:hatpiscalings}.
\end{align}

In \eqref{eq:hatpiscalings}, $\mathbf \Lambda$ is further partitioned in a block diagonal structure, with the blocks being themselves diagonal, conformally with the finer structure of $\cB$. Therefore, we write  $\varPi^\star$ as
\begin{align*}
\begin{array}{cccccc}\bigg[
\mathbf D_0B_1\mathbf \Lambda_1, & \mathbf D_0A_1 
\overbrace{\begin{bmatrix}
    B_2\mathbf \Lambda_2,& A_2B_3\mathbf \Lambda_3,& \ldots, & (A_2\cdots A_t)
\end{bmatrix}}^{\cC_{2:t+1}}\bigg].
\end{array}
\end{align*}
Our goal is to show that  $\varPi^\star$ arises as a transition probability of a Markov law. For this reason, in the above listing of the block entries of $\varPi^\star$, we see a nested structure and accordingly define $\cC_{2:t+1}$ as in the displayed equation above. 

The scalings in Proposition \ref{prop:partial} ensure the solution is row stochastic, i.e., $\varPi^\star \mathds 1 = \mathds 1$. 
It is important to note that $\cC_{2:t+1}$ may not be row stochastic. Therefore, we determine
the diagonal scaling of $\cC_{2:t+1}$ to ensure it becomes row stochastic.
To this end, let $\cC_{2:t+1}\mathds 1 =  D_1$, which is not necessarily equal to $\mathds 1$. Then, let  $\mathbf D_1={\rm diag}(D_1)^\sharp$, where $^\sharp$ denotes the Moore-Penrose generalized inverse that coincides with the standard inverse if no entries of $D_1$ vanish and rewrite $\varPi^\star$ as
\begin{align*}
\begin{array}{cccccc}\bigg[
\mathbf D_0B_1\mathbf \Lambda_1, & \mathbf D_0A_1\mathbf D_1^{\sharp} 
\overbrace{\begin{bmatrix}
    \mathbf D_1 B_2\mathbf \Lambda_2,& \ldots, & \mathbf D_1 A_2\cdots A_t
\end{bmatrix}}^{\mathbf D_1\cC_{2:t+1}}\bigg].
\end{array}
\end{align*}
Then, the following matrix is row stochastic:
\begin{align*}
\begin{bmatrix}
\mathbf D_0 B_1\mathbf \Lambda_1, & \mathbf D_0 A_1\mathbf D_1^{\sharp} \end{bmatrix}.
\end{align*}\

In the same way, by considering this telescopic expansion, we construct the sequence of $n \times n$ diagonal scalings $\mathbf D_1,\mathbf D_2,\ldots, \mathbf D_t $ via
\begin{align}\label{eq:ds}
\mathbf D_\tau={\rm diag}(D_\tau)^{\sharp} \ \ \text{ for } \ \ D_\tau=\cC_{\tau+1:t+1}\mathds{1}
\end{align}
so that
\begin{align*}
\varPi^\star = \bigg[
      \underbrace{B_1^\star, \ \ A_1^\star B_2^\star, \ \  \ldots, \ \  (A_1^\star \cdots  A_{t-1}^\star)  B_t^\star}_{\cB^\star},  \ \   \underbrace{(A_1^\star \cdots A_t^\star)}_{\cA^\star}
\bigg], 
\end{align*}
with
\begin{equation}\label{eq:hats}
\begin{split}
    B^\star_\tau &= \mathbf D_{\tau-1}B_\tau\mathbf \Lambda_\tau, \ \ \text{ and } \ \     A^\star_\tau =\mathbf D_{\tau-1}A_\tau \mathbf D_\tau^{\sharp},
\end{split}
\end{equation}
for all $\tau\in [1,t]$.  For completeness, we herein list the sequence of matrices used in \eqref{eq:ds}:
   \begin{align} \label{eq:cs}
\begin{split}
 \cC_{2:t+1} &= \begin{bmatrix}
    B_2\mathbf \Lambda_2, & A_2 B_3\mathbf \Lambda_3, & \ldots, & (A_2\cdots A_t)
\end{bmatrix}, \\
\cC_{3:t+1}  &= \begin{bmatrix} B_3 \mathbf \Lambda_3, & \ldots, & (A_3\cdots A_t)
\end{bmatrix},\\
& \  \vdots  \\
\cC_{t+1:t+1} &= I.
\end{split}
\end{align}%
Equations  \eqref{eq:ds}-\eqref{eq:cs} provide the parameters of the sought solution
to Problem \ref{prob:KL-paths-mult}. We first recap that these provide a Markov law that is consistent with the marginals, and in the next section, we argue that it, in fact, minimizes
the likelihood functional in \eqref{eq:likelihoodfn}.

\vspace{4pt}
\begin{proposition}\label{prop:structureoflaw}
       The Markov law \eqref{eq:MLaw1} with transition kernel \eqref{eq:postmat} such that $\hat B_\tau=B_\tau^\star$, and $\hat A_\tau=A_\tau^\star$ as in \eqref{eq:hats}
       is consistent with the given marginals $\hat \mu_0, \{\hat \nu_1, \ldots,\hat \nu_m\}$  of Problem \ref{prob:KL-paths-mult}. \hfill{$\Box$}
\vspace{4pt}
\end{proposition}

\begin{proof}
    The proceeding arguments in this section have established the proposition. Indeed, $\varPi^\star$ is of the required form
    \begin{align*}
     \begin{bmatrix}
        B^\star_1, & A_1^\star B_2^\star, & \ldots, & (A_1^\star\cdots A_{t-1}^\star)B_t^\star,& (A_1^\star \cdots A_t^\star) \end{bmatrix},
    \end{align*}
    with $\begin{bmatrix}
        B_\tau^\star, & A_\tau^\star \end{bmatrix}$ row stochastic for all $\tau \in [1,t]$, then the law with the kernel 
        \begin{align*}
            \Pi_\tau^\star = \begin{bmatrix}
                I & \mathbf 0 \\
                 B_\tau^\star, & A_\tau^\star
            \end{bmatrix}
        \end{align*}
        that forms $\varPi^\star$ satisfies \eqref{eq:constr1}. By taking  $\hat \mu$ as in \eqref{eq:mu}, and from \eqref{eq:hats} and Proposition \ref{prop:partial}, we deduce
    \begin{align*}
    \hat\mu^\intercal\begin{bmatrix}
        B^\star_1, & A_1^\star B_2^\star, & \ldots,& (A_1^\star \cdots A_t^\star) \end{bmatrix}=\hat\nu^\intercal,
    \end{align*}
    with $\hat \nu^\intercal$ is that of \eqref{eq:v} as required.  As a consequence, the law with the kernel forming $\varPi^\star$ also satisfies \eqref{eq:constr2}.
\end{proof}

\subsection{Optimality of the Law}\label{sec:III}

We are now in a position to establish our central result. That is, the Markov law identified in Proposition \ref{prop:structureoflaw} is, in fact, the optimal solution claimed in Theorem \ref{thm:thm2} and sought for in Problem \ref{prob:KL-paths-mult}. It is important to note that so far, we have only established that if a law has a kernel coinciding with \eqref{eq:hats}, it is consistent with the given marginals. We still need to show how such a law weighs in on individual paths and that it is optimal in that it minimizes the relative entropy functional in \eqref{eq:SB1} or \eqref{eq:likelihoodfn}. As it turns out, this is a consequence of the diagonal structure of scalings.

The key to seeing this is to observe first that the prior law $Q$ and
the posterior $P^\star$ share the same  pinned (Brownian) bridges, i.e.,
 the probability of a stopped path (recall \eqref{eq:setofpaths}) conditioned on a start at some $x_0 \in \Vt$ and on first absorption at some $x_\tau\in \Va$ is the same under $Q$ or $P^\star$. Likewise, the probability of an unstopped path conditioned on a start at some $x_0 \in \Vt$ and on an end at some $x_t\in \Vt$ is the same under either law. This is shown next.

\vspace{4pt}
\begin{proposition}\label{prob:disint1}
      Consider the Markov laws $Q$ and $P$ in \eqref{eq:priorlaw1} and \eqref{eq:MLaw1} with transition kernels \eqref{eq:priormat} and \eqref{eq:postmat}, respectively. If the law $P$ is such that $\hat B_\tau=B^\star_\tau$, and $\hat A_\tau=A^\star_\tau$ as in \eqref{eq:hats}, then 
        \begin{align*}
           &\! \! P(\mathbf x \mid X_0 = x_0,X_T=j, T=\tau) =  \nonumber \\
            &\hspace{70pt} Q(\mathbf x \mid X_0 = x_0,X_T=j, T=\tau), \\
          &\text{and } P(\mathbf x \mid X_0 = x_0, X_t =x_t) =  Q(\mathbf x \mid X_0 = x_0, X_t =x_t),
        \end{align*}
        for any $\mathbf x \in \mathcal V^{t+1}$, $x_0, x_t \in \Vt, j \in \Va$ and $\tau \in [1,t]$.  \hfill{$\Box$}
\end{proposition}
\vspace{4pt}

\begin{proof}
    In general,
    \begin{align}\label{eq:Qconditioned}
    & Q(\mathbf x \mid X_0 = x_0,X_T=j, T=\tau)= \nonumber \\
    &\hspace{70pt} \frac{1}{Z}   \Pi_1(x_0,x_1) \cdots \Pi_\tau (x_{\tau-1},j)
    \end{align}
    where $x_0, \ldots, x_{\tau-1} \in \Vt, j \in \Va, \tau \in [1,t]$, and 
    \begin{align*}
        Z = \sum_{x_1,\ldots,x_{\tau-1} \in \Vt} \Pi_1(x_0,x_1)\cdots \Pi_\tau(x_{\tau-1},j)
    \end{align*}
is a normalizing factor, assumed nonzero.
Bringing in the form of the kernels in \eqref{eq:priormat}, we have that $Q(\mathbf x \mid X_0 = x_0,X_T=j, T=\tau)$ from \eqref{eq:Qconditioned} is proportional to
\begin{align*}
A_1(x_0,x_1)\cdots A_{\tau-1}(x_{\tau-2},x_{\tau-1})B_\tau(x_{\tau-1},j).
\end{align*}

From \eqref{eq:hats}, we obtain that $P(\mathbf x \mid X_0 = x_0,X_T=j, T=\tau)$ is proportional to
\begin{align*}
D_0(x_0) A_1(x_0,x_1)\!\cdots\! A_{\tau-1}(x_{\tau-2},x_{\tau-1})B_\tau(x_{\tau-1},j)\Lambda_\tau(j).
\end{align*}
Since $x_0$ and $j$ are fixed, $D_0(x_0)\Lambda_\tau(j)$ is independent of $(x_1,\ldots,x_{\tau-1})$. Thus, the probability of a path starting from $x_0$ and first reaching $\Va$ at $\tau$ and at a particular $j$ is identical under the two laws.

Similarly, if we pin the bridge to some values $x_0\in\Vt$ and $x_t\in\Vt$ at the two endpoints, the probability mass on any particular path, while also a function of the path, is once again the same under the two laws. The verification is identical.
\end{proof}
\vspace{4pt}
We proceed with the section by stating and proving our central result.\\

\begin{theorem}
    The Markov law \eqref{eq:MLaw1} with transition kernel \eqref{eq:postmat} such that $\hat B_\tau=B_\tau^\star$, and $\hat A_\tau=A_\tau^\star$ as in \eqref{eq:hats} 
is the unique solution $P^\star$ to Problem \ref{prob:KL-paths-mult} assuming feasibility.     \hfill{$\Box$}
\end{theorem}
\vspace{4pt}

\begin{proof}
To establish the optimality with respect to the relative entropy functional, the basic idea of the proof is to disintegrate each probability distribution on paths as a sum of two products. The joint probability of an initial state, an absorption state, and associated stopping time appears in one term, whereas the joint probability of transient initial and final states appears in the other.
The relative entropy, upon substitution, decomposes into summands involving the respective terms of disintegrated measures.
When minimizing the relative entropy, two summands have already been optimized for due to our selection of $B_\tau^\star$ and $A_\tau^\star$, whereas others vanish by Proposition \ref{prob:disint1}. We now explain this.
To this end, we disintegrate the measure $Q$ and write
 \begin{align*}
 Q(\cdot)& = Q(\cdot \mid X_0 = x_0,X_T=j, T=\tau) \mathcal Q(x_0,j,\tau)\\
 & + Q(\cdot \mid X_0 = x_0,X_t=x_t) \mathcal Q(x_0,x_t)
\end{align*}
 where $\mathcal Q(x_0,j,\tau),\mathcal Q(x_0,x_t)$ are joint probabilities defined as
 \begin{align}
     \mathcal Q(x_0,j,\tau)& := Q(X_0 = x_0,X_T=j, T=\tau), \label{eq:j1} \\
  \text{and }   \mathcal Q(x_0,x_t) &:=  Q(X_0 = x_0,X_t=x_t), \label{eq:j2}
 \end{align}
 for any $x_0, x_t \in \Vt, j \in \Va$ and $\tau \in [1,t]$.
 
 The same applies to $P$ with $\mathcal P(x_0,j,\tau),\mathcal P(x_0,x_t)$ being the respective counterparts to \eqref{eq:j1},\eqref{eq:j2}. It is straightforward to verify that  
\begin{align*}
\mathbb D(P \parallel  Q) &=\sum_{\substack{x_0 \in \Vt, j \in \Va, \tau \in [1,t]}}
  \mathcal P(x_0,j,\tau) \log \frac{\mathcal P (x_0,j,\tau)}{\mathcal Q (x_0,j,\tau)}  \\
 &+ \sum_{x_0,x_t \in \Vt} \mathcal P(x_0,x_t) \log \frac{\mathcal P (x_0,x_t)}{\mathcal Q (x_0,x_t)}
 \end{align*}
when Proposition \ref{prob:disint1} holds and by noting

\begin{align*}
    \sum_{\mathbf x \in \bigcup\limits_{\tau \in [1,t],{j \in \Va}}  \! \! \!\!\!\! \Xt}  \! \! \! \! \!  \! \! \! Q(\mathbf x \mid X_0 = x_0,X_T=j, T=\tau) =1, \\
   \text{and }  \sum_{\mathbf x \in \Vtpt}  Q(\mathbf x \mid X_0 = x_0,X_t=x_t) =1.
\end{align*}%
Therefore, the relative entropy between the two laws is dictated only by joint probabilities.
Here, we note that if the kernel of $P$ is such that $\hat B_\tau=B_\tau^\star$, and $\hat A_\tau=A_\tau^\star$ as in \eqref{eq:hats}, then
\begin{align*}
\mathbb D(P \parallel  Q)&= \sum_{x,y}
 \hat\mu(x) \varPi^\star(x,y) \bigg [\log \frac{\hat\mu(x)}{\mu(x)} + \log \frac{\varPi^\star(x,y)}{\varPi(x,y)} \bigg],   
\end{align*}
where the summations are carried over all the indices and $\mu(x) = \mu_0(x)$ for $x \in \Vt$. Since $\varPi^\star(x,y)$ is row stochastic, $\mathbb D(P \parallel  Q)$ equals
\begin{align*}
 \sum_{x}
 \hat\mu(x) \log \frac{\hat\mu(x)}{\mu(x)} +\sum_{x,y} \hat \mu(x) \varPi^\star(x,y)\log \frac{ \varPi^\star(x,y)}{\varPi(x,y)}.
\end{align*}
The first term is independent of our choice of the kernel, whereas the second is optimal by Proposition \ref{prop:partial}, attesting optimality in the context of Problem \ref{prob:partialSB}  to seek transition probabilities to match the space-time marginals. This completes the proof. 
\end{proof}

\vspace{4pt}

\begin{remark}
    Throughout, we assume feasibility when dealing with optimization problems. To our interest, the feasibility of Problem \ref{prob:SBwithST} or equivalently Problem \ref{prob:KL-paths-mult} requires that a suitable number of walkers stationed initially at various vertices have a path(s) to each absorbing vertex within some number of steps that would allow matching the stopping-time marginals. 
    This guarantees scheduling the transfer of the right amount of mass to the absorbing vertices according to the temporal marginals. A sufficient condition for feasibility on the prior graph topology is that the transient vertices are always adjacent to the absorbing ones, and the set $\Vt$ forms a strongly connected component. This, indeed, prevents the prior telescopic expansion from having zero elements. That is to say, no diagonal matrices $\mathbf D_0,\mathbf \Lambda$ can scale the zero elements of $\varPi$ to achieve nontrivial elements of the posterior telescopic expansion.
\hfill{$\square$}
    \end{remark}

 \section{Regularized transport on Graphs} \label{sec:duality}

The variational formalism, where in the spirit of Schr\"odinger's bridges, we minimize the relative entropy, is fairly flexible and can incorporate the cost of transportation into our model. 
Specifically, suppose that we would like to select a Markov policy for guiding walkers on a graph so as to minimize the average transportation cost in their journey between their starting distribution and their respective destinations. As in the setting of the earlier framework, their marginals at the end of the journey can be both in space and time as they reach absorbing states with specified arrival-time probabilities.

Now, assume that the cost of traversing an edge $(x_{\tau-1},x_{\tau})$ at time $\tau$ is known and quantified by a function $U_{\tau}(x_{\tau-1},x_{\tau})$. Evidently, the cost of transporting along $\mathbf x=(x_0,\ldots,x_t)$ is
\begin{align*}
U(\mathbf x):=\sum_{\tau \in [1,t]}U_\tau(x_{\tau-1},x_{\tau}),
\end{align*}
and the cumulative cost of transporting with a distribution $P$ of paths of random walkers over a window in time $[0,t]$ is
\begin{align*}
\mathbb J(P)&=\sum_{\mathbf x}
P(\mathbf x)U(\mathbf x),
\end{align*}
where the summation runs over paths that satisfy specifications.
As explained in \cite{chen2016robust,chen2021stochastic}, there are often practical advantages in modifying the cost functional of our control problem by adding an entropic regularization term, e.g., to minimize
\begin{align*}
\mathbb J(P)+\beta^{-1}\sum_{\mathbf x}
P(\mathbf x)\log P(\mathbf x).
\end{align*}
The constant $\beta^{-1}$ is thought of as ``temperature" and helps match the units\footnote{In statistical physics, $\beta^{-1}=k_BT$, i.e., the product of the Boltzmann constant and the temperature.} to those of $U$. From another angle, increasing temperature promotes randomness in selecting alternative paths and thereby trades off cost for robustness \cite{chen2016robust}.

More generally, we may choose to penalize the entropic distance (relative entropy) of a law $P$ from a given prior law $Q$. As before, $Q$ may be available as a point of reference. Thus, we may
adopt a ``free energy'' functional,
\begin{align*}
\mathbb F(P \parallel Q):=\mathbb J(P)+\beta^{-1}\mathbb D(P \parallel Q),
\end{align*}
to replace $\mathbb J$ as our optimization functional and seek a solution to the following problem.
\vspace{4pt}
\begin{problem}\label{prob:last}
     Consider a prior Markov law $Q$ on paths in $\mathcal V^{t+1}$ is given as in \eqref{eq:priorlaw} together with a probability vector $\hat\mu_0  \in \mathbb R_+^{n+m}$ with support in $\Vt\subset \mathcal V$ and a set of vectors 
    $\{\hat \nu_j  \in \mathbb R_+^{t} \mid j \in  [1,m]\}$
    as stopping-time marginals at $\Va$ such that \eqref{eq:ineq} is satisfied.
        Determine the law
        \begin{align}\label{eq:likelihoodfn2}
       P^\star  := \arg \min_{P \ll Q}  \mathbb F(P \parallel  Q),
        \end{align}
   subject to \eqref{eq:constr1} and \eqref{eq:constr2}.  \hfill{$\square$}
\end{problem}
\vspace{4pt}

Observing that
\begin{align*}
\mathbb J(P) = -\beta^{-1}\sum_{\mathbf x}
P(\mathbf x) \log e^{-\beta U(\mathbf x)},
\end{align*}
we can write that
\begin{align*}
\mathbb F(P \parallel  Q) = \mathbb D(P \parallel  Qe^{-\beta U}).
\end{align*}
This is identical in form to the relative entropy, albeit the quantity $Qe^{-\beta U}$ that combines cost and prior is a measure that may not be normalized (hence, not necessarily a probability measure). Yet, all the steps in our previous analysis carry through, and the law $P^\star$ to optimize this mix of cost and divergence from the nominal $Q$ can be similarly obtained. The new prior $Qe^{-\beta U}$ inherits the Markovian structure from $Q$ since $e^{-\beta U}$ factors similarly as
\begin{align*}
e^{-\beta U}&=f^{\rm cost}_1(x_0,x_1)\times \cdots \times f^{\rm cost}_t(x_{t-1},x_t),
\end{align*}
for $f^{\rm cost}_\tau(x_{\tau-1},x_\tau) = e^{-\beta U_\tau(x_{\tau-1},x_\tau)}$. Therefore, the posterior $P^\star $ that is obtained by minimizing $\mathbb F (P \parallel  Q)$ subject to constraints on spatio-temporal marginals as before is also Markov. We summarize it as follows.
\vspace{4pt}
\begin{theorem}
    Assuming that Problem \ref{prob:last} is feasible, the solution $P^\star$ is unique, and it is Markov. \hfill{$\square$}
\end{theorem}
\vspace{4pt}

\begin{remark}
The construction of the solution to Problem \ref{prob:last} proceeds precisely as in the earlier sections where the prior is normalized. Moreover, the solution to Problem \ref{prob:last} $P^\star(\mathbf x;\beta)$ recovers the solution for minimizing $\mathbb J$ subject to the same constraints as $\beta^{-1}\to 0$. Entropic regularization, besides its practical significance, is often used in computational tools since minimizing $\mathbb F$ is more efficient than minimizing $\mathbb J$ and thus a standard practice in the recent literature of optimal mass transport and machine learning \cite{chen2016entropic,peyre2019computational}. \hfill{$\square$}
\end{remark}

\section{Examples}\label{sec:examples}
Below, we first return to our motivating example of reconciling wins and losses in a martingale game, where unlikely stopping-time marginals may reveal a foul play. Next, we discuss a second academic example of scheduling probabilistic flow on a network so as to regulate the capacity at critical nodes. This example setting is framed as a transportation network of some city neighborhoods. The city has two one-way roads linking these neighborhoods to the downtown, and passing through either road is the stopping rule. We then assume that one of the roads closes during peak traffic hours and seek to redirect the flow with a Markov policy so that it abides by a prescribed stopping-time marginal for the traffic going across to meet the capacity constraints.

\subsection{De Moivre's Martingale}

We now reexamine the betting game in Section \ref{sec:martingale}, abstracted via  Fig. \ref{fig:fig1} and the graph in Fig. \ref{fig:DeMoivre}. 
Consider the players start the game with $1$ to $4$ tokens. Each state is represented as a node. In addition, node $0$ represents ruin, while node $5$ represents the profit cap, and reaching either node triggers exit from the game. Thus, nodes $0$ and $5$ are absorbing. Assume that players' wealth is distributed uniformly over the nodes $\{1,2,3,4\}$ at the start of the game.
 By letting the prior be that of a fair game, the prior transition probabilities are shown in Fig. \ref{fig:DeMoivre}.
\begin{figure}[H]
    \centering
    \includegraphics[width=\columnwidth]{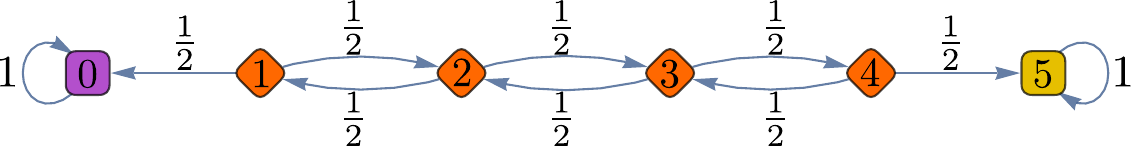}
    \caption{Network topology with prior transition probabilities.}
    \label{fig:DeMoivre}
\end{figure}
\noindent We assume that we observe three consecutive rounds. The portion of players exiting the game due to ruin or success is recorded, which gives the stopping-time marginals 
\begin{align*}
    \hat \nu^{\rm ruin} = \begin{bmatrix}
        1/8 & \textBF{1/5} & 1/16 
    \end{bmatrix}, \text{ and   }
    \hat \nu^{\rm win} = \begin{bmatrix}
        1/8 & 1/16 & 1/16 
    \end{bmatrix}.
\end{align*}

While the fair game should result in $\hat \nu^{\rm ruin} = \hat \nu^{\rm win}$, we notice an unexpectedly high percentage of players exiting in the second round due to ruin. We seek to pinpoint possible cheating that could explain the marginals. To this end, we determine the most likely transition kernel $(\Pi_\tau^\star)_{\tau \in [1,3]}$. By \eqref{eq:Sinkhorn}-\eqref{eq:cs}, we obtain the sequence
   \begin{align*}
 \Bigg\{      B^\star_1 &= \begin{bmatrix}
        \phantom{0}0.5\phantom{0}  & 0 \\
        0 & 0 \\ 
        0 & 0 \\
        0 & \phantom{0}0.5\phantom{0}
    \end{bmatrix} ,
    A^\star_1 = \begingroup
\setlength\arraycolsep{2pt} \begin{bmatrix}
 0 & \phantom{0}0.5\phantom{0} & 0 & 0 \\
 \textBF{0.858} & 0 & \textBF{0.142} & 0 \\
 0 & \phantom{0}0.5\phantom{0} & 0 & \phantom{0}0.5\phantom{0} \\
 0 & 0 & \phantom{0}0.5\phantom{0} & 0 \\
    \end{bmatrix},  \endgroup \\
      B^\star_2 &=
    \begin{bmatrix}
 \textBF{0.933} & 0 \\
 0 & 0 \\
 0 & 0 \\
 0 & \phantom{0}0.5\phantom{0}\hspace*{1pt} \end{bmatrix}  ,
    A^\star_2 =  \begingroup
\setlength\arraycolsep{2pt} \begin{bmatrix}
 0 & \textBF{0.067} & 0 & 0 \\
 \phantom{0}0.5\phantom{0}\hspace*{2pt} & 0 & \phantom{0}0.5\phantom{0} & 0 \\
 0 & \textBF{0.407} & 0 & \textBF{0.593} \\
 0 & 0 & \phantom{0}0.5\phantom{0} & 0 \\
    \end{bmatrix}, \endgroup \\
B^\star_3 &= \begin{bmatrix}
 \phantom{0}0.5\phantom{0} & 0 \\
 0 & 0 \\
 0 & 0 \\
 0 & \textBF{0.657} 
\end{bmatrix},
A^\star_3 =  \begingroup
\setlength\arraycolsep{2pt} \begin{bmatrix}
 0 & \phantom{0}0.5\phantom{0} & 0 & 0 \\
 \phantom{0}0.5\phantom{0} & 0 & \phantom{0}0.5\phantom{0} & 0 \\
 0 & \phantom{0}0.5\phantom{0} & 0 & \phantom{0}0.5\phantom{0} \\
 0 & 0 & \textBF{0.343} & 0 \\
 \end{bmatrix} \endgroup \Bigg\}
\end{align*}
that makes up the sought transition kernel. The result suggests that the game may have been staged in that 
with high probability ($80 \%$ chance),  players reach ruin in the second round, given that they start with $2$ tokens. This shows up in the value of $A_1^\star(2,1) \times B_2^\star(1,0)$. In the third round, the chances of winning may have been tampered with again so that the players with more tokens are more likely to win. This could have been arranged to make the game look fair in the third round.

 \subsection{ Congestion Control:  Traffic Flow Regulation}

  The graph in Fig.\ \ref{fig:Manhat} abstracts traffic flow in a city where nodes $1$ and $2$ represent two one-way roads, while nodes $3$ through $7$ represent neighborhoods around them. The two one-way roads connect the city to its downtown.
 We consider traffic flow during rush hours and assume that citizens cross either of the roads toward downtown during the time window at hand. Thus, we treat the two roads as absorption nodes and eventual destinations for the traffic.
 
  With a plan to repair one of the roads, the city hall aims to redirect traffic for three hours of scheduled maintenance. During this time, traffic must not exceed the capacity of the operating road to avoid traffic jams. We take this capacity value as $0.16$ of the average daily number of cars moving in the city. We model the problem on a network with two absorbing states. In that, $\Va = \{1,2\}$ with node $1$ being the road to be mended and $2$ being the one to remain operational.
   \begin{figure}[H]
    \centering
    \includegraphics[width=\columnwidth]{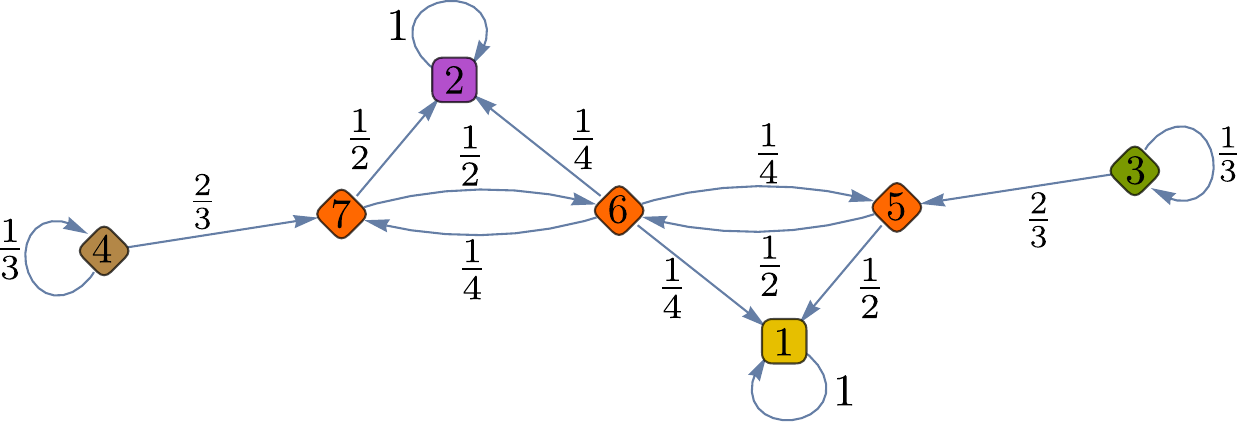}
    \caption{Network representation of vital neighborhoods in a city and associated transition probabilities under normal operation of the city roads.} 
    \label{fig:Manhat} 
\end{figure}%
  Typical flow transitions that serve to define the prior are also shown in Fig. \ref{fig:Manhat}.  
  During times $\tau\in\{3,4,5\}$ of maintenance, a possible policy is to adjust the prior kernel, redirecting traffic toward the road that remains open. Then, the prior is completed by taking
 \begin{align*}
 \Pi_3 = \Pi_4 = \Pi_5 = \left[ \begin{array}{c|c}
     \begin{matrix}
         1 &  0 \\
         0 & 1 \\
     \end{matrix}  & \begin{matrix}
        0\phantom{0} &  & \cdots &  & \phantom{0}0 \\
        0\phantom{0} &  & \cdots &  & \phantom{0}0
     \end{matrix} \\ 
     \hline & \\[\dimexpr-\normalbaselineskip+2pt]
      \begin{matrix}
     0&0\\
     0&0\\
     0&0\\
     0&\frac12\\[1.5pt]
     0&\frac12 
     \end{matrix} & \begin{matrix}
    \frac13&0&\frac23&0&0\\
     0&\frac13&0&0&\frac23\\
     0&0&0&1&0\\[1.5pt]
     0&0&\frac14&0&\frac14\\[1.5pt]
     0&0&0&\frac12&0
     \end{matrix}
     \end{array} \right].
 \end{align*}
 
 Let the initial probability vector be such that
 \begin{align*}
     \hat \mu_0(3:7) = \begin{bmatrix}
          0.4 &  0.4 &  0.05 & 0.1 & 0.05 
     \end{bmatrix}^\intercal.
 \end{align*}
With the closing of one of the roads, the prior temporal marginals become
\begin{align*}
  \nu_1 &= \begin{bmatrix}
     0.05 & 0.158 &	0 &	0  & 0 
  \end{bmatrix},\\
 \text{ and }  \nu_2 &=  \begin{bmatrix}
      0.05 & 0.158 &	\textBF{0.197} &	0.127  & 0.101
  \end{bmatrix}.
\end{align*}
We seek to adjust the flow so as not to exceed capacity and thereby target obtaining the time marginals
\begin{align*}
    \hat \nu_1 &=   \begin{bmatrix}
        0.05 & 0.158 &	0 &	0  & 0 
    \end{bmatrix},\\
   \text{and } \hat \nu_2 &=  \begin{bmatrix}
        0.05 & 0.158 &	0.142 &	0.142  & 0.142
    \end{bmatrix}.
    \end{align*}
    Indeed, this is achieved by solving  Problem \ref{prob:KL-paths-mult} to obtain the optimal kernel via \eqref{eq:Sinkhorn}-\eqref{eq:cs}, giving
\begin{align*}
  &  D_0 = \begin{bmatrix}
    0.975 \\
    1.034 \\
    1.008 \\
    1.024 \\
    1.008
    \end{bmatrix} , 
     \begin{bmatrix}
        \Lambda_1^\intercal \\
        \Lambda_2^\intercal\\
        \Lambda_3^\intercal\\
        \Lambda_4^\intercal \\
        \Lambda_5^\intercal
    \end{bmatrix} =\begin{bmatrix}
    0.984 & 0.984 \\
    1.019 & 0.97 \\
    0 & 0.711 \\
    0 & 1.108 \\
    0 & 1.4
    \end{bmatrix}, 
    D_1 = \begin{bmatrix}
    1.114 \\
    0.987 \\
    0.981 \\
    1 \\
    0.957
    \end{bmatrix}, \\
   & D_2 = \begin{bmatrix}
    1.133 \\
    1.147 \\
    1.104 \\
    0.944 \\
    0.908
    \end{bmatrix} , 
    D_3 = \begin{bmatrix}
    1 \\
    1.133 \\
    1.2 \\
    1.104 \\
    1.154 
    \end{bmatrix}, 
     D_4 = \begin{bmatrix}
1 \\
1 \\
1 \\
1.2 \\
1.2 \\
    \end{bmatrix}, 
 \text{and }   D_5 = \begin{bmatrix}
        1 \\ \vdots \\ 1
    \end{bmatrix}.
\end{align*}

\section{Concluding remarks}\label{sec:conclusions}

It is natural to consider analogous problems, such as the instance of a space probe at landing, where a process is stopped when specific targets are reached and stopping criteria are met. Such problems are of evident practical significance.
The present work shows that similar types of problems can be treated within the theory of Schr\"odinger's bridges in discrete time and space. It is envisioned that the possibility to tackle spatio-temporal soft conditioning in several control problems will open a new phase in the developing topic of uncertainty control, complementing recent works such as \cite{chen2019relaxed,chen2022most,tsiotras,bakolas2018finite,caluya2021wasserstein,pavon2021data}.

\balance
\bibliographystyle{ieeetr}
\bibliography{main.bib}

\end{document}